\documentclass[conference]{IEEEtran}
\usepackage[british]{babel}

\usepackage[noadjust]{cite}

\usepackage{array}
\usepackage{stfloats}
\usepackage{url}

\usepackage{amsthm,amssymb}
\usepackage{amsmath}
\usepackage{mathtools}

\usepackage{tikz}
\usepackage{enumerate}
\usepackage{xcolor}
\usepackage{graphicx}

\usetikzlibrary{matrix}

\usepackage{ebproof}

\usepackage{nicefrac}
\newcommand{\sfrac}[2]{\nicefrac{#1}{#2}}

\allowdisplaybreaks

\def\emph{\textbf}

\newcommand{\ie}{i.e.~}
\newcommand{\eg}{e.g.~}

\theoremstyle{plain}
\newtheorem{theorem}             {Theorem}
\newtheorem{proposition}[theorem]{Proposition}

\newtheorem*{theorem*}    {Theorem}
\newtheorem*{proposition*}{Proposition}
\newtheorem*{lemma*}      {Lemma}
\newtheorem*{corollary*}  {Corollary}
\newtheorem*{conjecture*} {Conjecture}

\theoremstyle{definition}
\newtheorem{definition}[theorem]{Definition}
\newtheorem{example}   [theorem]{Example}
\newtheorem*{definition*}{Definition}
\newtheorem*{example*}   {Example}

\theoremstyle{remark}

\newcommand{\Mcomma}{\text{,}}
\newcommand{\Msemicolon}{\text{;}}
\newcommand{\Mdot}{\text{.}}

\usepackage{colonequals}
\newcommand{\defeq}{\colonequals} 
\newcommand{\tuple}[1]{\mathopen{\langle}#1\mathclose{\rangle}} 
\newcommand{\setdef}  [2]{\left\{#1 \mid #2\right\}}             
\newcommand{\enset}   [1]{\mathopen{ \{ }#1\mathclose{ \} }} 
\newcommand{\family}   [1]{\mathopen{ ( }#1\mathclose{ ) }} 
\newcommand{\fdec}    [3]{#1 \colon #2 \longrightarrow #3}
\newcommand{\fdef}    [3]{#1 \coloncolon #2 \longmapsto     #3}
\newcommand{\fdecdef} [5]{#1 \colon #2 \longrightarrow #3 \coloncolon #4 \longmapsto #5}
\newcommand{\anotfdec}[4]{#2: #3 \overset{#1}{\longrightarrow} #4}
\newcommand{\isofdec} [3]{\anotfdec{\cong}{#1}{#2}{#3}}

\newcommand{\id}[1][]{\ensuremath{\mathrm{id}_{#1}}}
\newcommand{\cat}[1]{\ensuremath{\mathbf{#1}}}
\newcommand{\lto}{\longrightarrow}
\newcommand{\longhookrightarrow}{\lhook\joinrel\rightarrow}

\newcommand{\Forall}[1]{\forall {#1}\boldsymbol{.}\;}

\def\implies{\Rightarrow}

\newcommand{\RRpz}{\mathbb{R}_{\geq 0}}

\newcommand{\MP}{\mathsf{MP}}
\newcommand{\lk}{\mathrm{lk}}
\newcommand{\NCF}{\mathsf{NCF}}
\newcommand{\CF}{\mathsf{CF}}
\newcommand{\x}{\bar{x}}
\newcommand{\y}{\bar{y}}
\newcommand{\z}{\bar{z}}
\newcommand{\simulates}{\rightsquigarrow}
\DeclareMathOperator{\supp}{\mathsf{supp}}

\newcommand{\choice}{\mathbin{\&}}

\newcommand{\X}{\mathbf{X}} 
\newcommand{\Y}{\mathbf{Y}} 
\newcommand{\Z}{\mathbf{Z}} 

\newcommand{\XSO}{\tuple{X,\Sigma,O}}
\newcommand{\XSOp}{\tuple{X,\Sigma,O'}}
\newcommand{\XpSpOp}{\tuple{X',\Sigma',O'}}

\newcommand{\XpSpfO}{\tuple{X',\Sigma',f^*O}}

\newcommand{\YTP}{\tuple{Y,\Theta,P}}

\newcommand{\Vertices}{\mathsf{Vert}}

\newcommand{\Dist}{\mathsf{D}}
\newcommand{\Set}{\mathsf{Set}}
\newcommand{\Ev}{\mathcal{E}}
\newcommand{\op}{\mathsf{op}}
\newcommand{\one}{\mathbf{1}}

\newcommand{\cok}{\dagger}

\newcommand{\itemd}[1]{\item\textbf{#1.} }

\usepackage{calc}
\newlength{\twoparindent}
\makeatletter 
\newenvironment{shiftflalign}{\@fleqntrue\@mathmargin\twoparindent\align}{\endalign\@fleqnfalse}
\makeatother

\newcommand{\nrule}[1]{\hphantom{(choice)}\llap{(#1)}}
\newcommand{\sidecon}[1]{{\scriptsize #1}}

\newcommand{\rvar}{\begin{prooftree}
  \infer[left label=\nrule{var}]{0}{\Gamma, \enset{v : \tuple{X,\Sigma,O}} \vdash v : \tuple{X,\Sigma,O}}
\end{prooftree}}
\newcommand{\rzero}{\begin{prooftree}
  \infer[left label=\nrule{zero}]{0}{\Gamma \vdash \ezero : \tuple{\emptyset,\Delta_0, \family{ }} }
\end{prooftree}}
\newcommand{\rsingl}{\begin{prooftree}
  \infer[left label=\nrule{singl}]{0}{\Gamma \vdash \eone : \tuple{\one,\Delta_1,\family{\one}}}
\end{prooftree}}
\newcommand{\rmeas}{\begin{prooftree}
  \hypo{\Gamma \vdash t : \XSO}
  \infer[left label=\nrule{meas}]1[\sidecon{$\fdec{f}{\Sigma'}{\Sigma}$ simplicial}]{\Gamma \vdash f^* t : \XpSpfO}
\end{prooftree}}
\newcommand{\routc}{\begin{prooftree}
 \hypo{\Gamma \vdash t : \XSO}
 \infer[left label=\nrule{outc}]1[\sidecon{$\family{\fdec{h_x}{O_x}{O'_x}}_{x \in X}$}]{\Gamma \vdash t/h : \XSOp}
\end{prooftree}}
\newcommand{\rmix}{\begin{prooftree}
 \hypo{\Gamma \vdash t : \XSO}
 \hypo{\Gamma' \vdash t': \XSO}
 \infer[left label=\nrule{mix}]2[\sidecon{$\lambda \in [0,1]$}]{\Gamma, \Gamma' \vdash t +_\lambda t' : \XSO}
\end{prooftree}}
\newcommand{\rchoice}{\begin{prooftree}
 \hypo{\Gamma \vdash t : \XSO}
 \hypo{\Gamma' \vdash t':\XpSpOp}
 \infer[left label=\nrule{choice}]2{\Gamma, \Gamma' \vdash t \choice t' : \tuple{X\sqcup X', \Sigma + \Sigma', [O,O']}}
\end{prooftree}}
\newcommand{\rprod}{\begin{prooftree}
 \hypo{\Gamma \vdash t : \XSO}
 \hypo{\Gamma' \vdash t':\XpSpOp} \infer[left label=\nrule{prod}]2{\Gamma, \Gamma' \vdash t \otimes t' : \tuple{X\sqcup X', \Sigma \star \Sigma', [O,O']}}
\end{prooftree}}
\newcommand{\rcond}{\begin{prooftree}
\hypo{\Gamma \vdash t : \XSO}
\infer[left label=\nrule{cond}]1[\sidecon{$x \in X$, $y = \family{y_o \in \Vertices(\lk_x\Sigma)}_{o \in O_x}$}]{\Gamma \vdash t[x?y] : \tuple{X \cup \enset{x?y},\Sigma[x?y],O]}}
\end{prooftree}}

\newcommand{\XYTableNUM}[4]{
\begin{array}{cc|cccc}
  \text{A} & \text{B} & 0\,0 & 0\,1 & 1\,0 & 1\,1 \\
  \hline
  x_1 & y_1 & #1 \\ 
  x_1 & y_2 & #2 \\ 
  x_2 & y_1 & #3 \\ 
  x_2 & y_2 & #4 \\ 
\end{array}
}
\newcommand{\XYTablePR}{
\XYTableNUM
{\sfrac{1}{2} &  0  &  0  & \sfrac{1}{2}}
{\sfrac{1}{2} &  0  &  0  & \sfrac{1}{2}}
{\sfrac{1}{2} &  0  &  0  & \sfrac{1}{2}}
{0 & \sfrac{1}{2} & \sfrac{1}{2} & 0}}

\newcommand{\PR}{\text{PR}}

\begin{document}
\title{A comonadic view of \\ simulation and quantum resources}
\author{\IEEEauthorblockN{Samson Abramsky\IEEEauthorrefmark{1},
Rui Soares Barbosa\IEEEauthorrefmark{1},
Martti Karvonen\IEEEauthorrefmark{2}, and
Shane Mansfield\IEEEauthorrefmark{3}}
\IEEEauthorblockA{\IEEEauthorrefmark{1}Department of Computer Science, University of Oxford, U.K.}
\IEEEauthorblockA{\IEEEauthorrefmark{2}School of Informatics, University of Edinburgh, U.K.}
\IEEEauthorblockA{\IEEEauthorrefmark{3}Laboratoire d'Informatique de Paris 6, CNRS and Sorbonne Universit\'e, France}}

\IEEEoverridecommandlockouts

\maketitle

\begin{abstract}
We study simulation and quantum resources
in the setting of
the sheaf-theoretic approach to contextuality and non-locality. Resources are viewed behaviourally, as empirical models.
In earlier work, a notion of morphism for these empirical models was proposed and studied. We generalize and simplify the earlier approach, by starting with a very simple notion of  morphism, and then extending it to a more useful one by passing to a co-Kleisli category with respect to a comonad of measurement protocols. We show that these morphisms capture notions of simulation between empirical models obtained via ``free'' operations in a resource theory of contextuality, including the type of classical control used in measurement-based quantum computation schemes.
\end{abstract}

\section{Introduction}

A key objective in the field of quantum information and computation is to understand the advantage which can be gained in information-processing tasks by the use of quantum resources.
While a range of examples have been studied, to date a systematic understanding of quantum advantage is lacking.

One approach to achieving such a general understanding is through \emph{resource theories} \cite{horodecki2013quantumness,abramsky2017contextual}, in which one considers a set of operations by which one system can be transformed into another. In particular, one considers ``free operations'', which can be performed without consuming any additional resources of the kind in question. 
If resource $B$ can be constructed from $A$ using only free operations, then we say that $A $ is convertible to $B$, or $B$ is reducible to $A$.
This point of view is studied in some generality in \cite{coecke2016mathematical,fritz2017resource}.

Another natural approach, which is familiar in computation theory, is to consider a notion of \emph{simulation}; one asks if the behaviour of $B$ can be produced by some protocol using $A$ as a resource.

Both these points of view can be considered in relation to quantum advantage. Our focus in this paper is on quantum resources that take the form of \emph{non-local}, or more generally \emph{contextual}, correlations. Contextuality is one of the key signatures of non-classicality in quantum mechanics \cite{ks,bell1966}, and has been shown to be a necessary ingredient for quantum advantage in a range of information-processing tasks \cite{raussendorf2013contextuality,howard2014contextuality,bermejo2017contextuality,abramsky2017contextual}.

In previous work \cite{abramsky2017contextual},
a subset of the present authors
showed how this advantage could be quantified in terms of the \emph{contextual fraction}, and also introduced a range of free operations, which were shown to have the required property of being non-increasing with respect to the contextual fraction.  Thus this work provided some of the basic ingredients for a resource theory of quantum advantage, with contextuality as the resource.

In \cite{karvonen2018categories}, the other present author introduced a notion of simulation between (possibly contextual) behaviours, as morphisms between empirical models, in the setting of the ``sheaf-theoretic'' approach to contextuality introduced in \cite{ab}. This established a basis for a simulation-based approach to comparing resources.

In this paper, we bring these two approaches together.
\begin{itemize}
\item On the simulation side, we enhance the treatment given in \cite{karvonen2018categories} by introducing a \emph{measurement protocols} construction on empirical models (Section~\ref{ssec:mps}). Measurement protocols were first introduced in a different setting in \cite{acin2015combinatorial}. This construction captures the intuitive notion, widely used in an informal fashion in concrete results in quantum information (e.g.~\cite{barrettpironio2005}),
 of using a ``box'' or device by performing some measurement on it, and then, depending on the outcome, choosing some further measurements to perform. This form of adaptive behaviour also plays a crucial role in measurement-based quantum computing \cite{raussendorf2001one}.

We show that this construction yields a comonad on the category of empirical models. Hence, we are able to describe a very general notion of simulation of $B$ by $A$ in terms of co-Kleisli maps from $A$ to $B$ (Sections~\ref{ssec:comonad} and \ref{ssec:generalsimulations}).

\item We consider the algebraic operations previously introduced in \cite{abramsky2017contextual} and introduce a new operation allowing a conditional measurement, a one-step version of adaptivity (Section~\ref{ssec:operations}).
We present an equational theory for these operations and use this to obtain normal forms for resource expressions (Section~\ref{ssec:eqtheory}).

\item Using these normal forms, we obtain one of our main results: we show that the algebraic notion of convertibility coincides with the existence of a simulation morphism (Section~\ref{ssec:generalsimulations}).

\item We also prove some further results, including a form of no-cloning theorem at the abstract level of simulations (Section~\ref{ssec:nocloning}).
\end{itemize}

\vfill

\section{Empirical Models}\label{sec:empiricalmodels}

We begin by introducing the main ingredients of the
sheaf-theoretic approach to contextuality~\cite{ab}.
The central objects of study are empirical models. These describe the behaviours that we are considering as resources, which may be contextual.

The behaviour intended to be modelled
is that of a physical system, governed perhaps by the laws of quantum mechanics,
on which one may perform measurements and observe their outcomes.
We abstract away from the details of the physical description of the system in question
and consider only its observable behaviour,
\ie the empirical distributions of such measurement experiments. 

We can therefore think of an empirical model as a black box,
with which an agent
might interact
by way of questions (measurements) and answers (outcomes).
The interface or type of such a box is given by a \emph{measurement scenario},
which specifies the allowed measurements and the
set of possible outcomes for each of them.
In a single use of the black box, the agent may perform multiple measurements.
However, a crucial feature that is typical in quantum systems is that some combinations of measurements may not be compatible.
In particular, it is typically not the case that the agent may jointly perform \textit{all} of the available measurements.
The scenario must, therefore, specify which sets of measurements are compatible and
can thus be performed together -- or sequentially in any order -- in a single use of the black box.
Sets of compatible measurements are called \emph{measurement contexts}.

This compatibility structure on measurements can be naturally described in terms of a simplicial complex.
Recall that an (abstract) \emph{simplicial complex} on $X$ is a set of finite subsets of $X$, called \emph{faces}, that is non-empty, downwards-closed in the inclusion order, and contains all the singletons. Concretely, these axioms amount to saying that any subset of a compatible set of measurements is a compatible set of measurements, and that any single measurement should be possible.

\begin{definition} 
A
\emph{measurement scenario} is a triple $\X = \XSO$
where:
\begin{itemize}
\item $X$ is a finite set of measurements;
\item $O = (O_x)_{x \in X}$ specifies, for each measurement $x \in X$, a finite non-empty set $O_x$ of  outcomes;
\item $\Sigma$ is a simplicial complex on $X$, whose faces are called the \emph{measurement contexts}.
\end{itemize}
\end{definition}

We will often simply refer to these as scenarios and contexts.
Note that a simplicial complex is determined by its maximal faces, called facets.
Hence, the measurement compatibility structure can be specified by providing only the maximal contexts, as was the case e.g. in \cite{ab}.

\begin{definition}
Let $\XSO$ be a scenario.
For any $U \subseteq X$, we write \[\Ev_O(U) \defeq \prod_{x \in U} O_x \]
for the set of assignments of outcomes to each measurement in the set $U$.
When $U$ is a valid context, these are the joint outcomes one might obtain for the measurements in $U$.
This extends to a sheaf  $\fdec{\Ev_O}{\mathcal{P}(X)^\op}{\Set}$,
with restriction maps $\fdec{\Ev_O({U\subseteq V})}{\Ev_O(V)}{\Ev_O(U)}$ given by the obvious projections.
We call this the \emph{event sheaf}. Whenever it does not give rise to ambiguity, we omit the subscript and denote the event sheaf more simply by $\Ev$.
\end{definition}

We write $\fdec{\Dist}{\Set}{\Set}$ for the functor of finitely-supported probability distributions.
For a set $S$,
\begin{equation*}
\begin{multlined}
\Dist(S) \defeq \\ \setdef{ \fdec{d}{S}{\RRpz} }{ \supp(d) \text{ is finite, } \sum_{s \in S}d(s) = 1 } \Mcomma
\end{multlined}
\end{equation*}
where $\supp(d) \defeq \setdef{s \in S}{d(s)\neq 0}$. The action of $\Dist$ on a  function $\fdec{f}{S}{T}$ is given by pushforward of distributions:
\[\fdecdef{\Dist(f)}{\Dist(S)}{\Dist(T)}{d}{\lambda t \in T. \sum_{s \in S, f(s)=t} d(s)} \Mdot\]
Note that, in particular, the pushforward $\Dist(\pi)$ along a projection ${\fdec{\pi}{S_1 \times S_2}{S_1}}$ corresponds to taking marginal distributions.
\begin{definition}
An \emph{empirical model} $e$ on a scenario $\XSO$, written $e : \XSO$,
is a compatible family for $\Sigma$ on the presheaf $\Dist \circ \Ev$.
More explicitly, it is a family $\family{e_\sigma}_{\sigma \in \Sigma}$ where,
for each $\sigma\in\Sigma$,
\[e_\sigma \in \Dist \circ \Ev(\sigma) = \Dist\left( \prod_{x \in \sigma}O_x \right)\]
is a probability distribution over the joint outcomes for the measurements in the context $\sigma$.
Moreover, compatibility requires that the marginal distributions be well-defined:
for any $\sigma, \tau \in \Sigma$ with
$\tau \subseteq \sigma$,
one must have 
\[e_{\tau} = e_{\sigma}|_\tau = \Dist \circ \Ev(\tau \subseteq \sigma)(e_\sigma) \Mcomma \]
\ie for any $t \in \Ev(\tau)$,
\[e_\tau(t) = \sum_{s \in \Ev(\sigma), s|_\tau = t} e_\sigma(s)\Mdot\]
\end{definition}
Note that compatibility can equivalently be expressed as the requirement that, for all facets (\ie maximal contexts) $C$ and $C'$ of $\Sigma$,
\[e_C|_{C\cap C'} = e_{C'}|_{C\cap C'} \Mdot\]
Compatibility holds for all quantum realizable behaviours \cite{ab}, and generalizes a property known as \emph{no-signalling} \cite{ghirardi1980general}, which we illustrate in the following example.

\begin{example}
Consider a bipartite black box shared between parties Alice and Bob, each of whom may choose to perform as their input one of two measurements.
We call Alice's measurements $x_1$ and $x_2$ and Bob's measurements $y_1$ and $y_2$.
Each measurement outputs an outcome that is either $0$ or $1$.
The situation can be described by a measurement scenario $\XSO$ in which $X=\enset{x_1,x_2,y_1,y_2}$, $O_x = \enset{0,1}$ for all $x \in X$, and the facets of $\Sigma$ are
\[ \enset{ \, \enset{x_1,y_1}, \, \enset{x_1,y_2} , \enset{x_2,y_1}, \enset{x_2,y_2} \, } \Mdot \]

The probabilistic behaviour of such a black box could be given e.g. by Table~\ref{tab:example}.
This happens to show a well-studied behaviour known as a Popescu--Rohrlich (PR) box \cite{popescu1994quantum}.
Rows of this table correspond to maximal measurement contexts, and columns to their joint outcomes.
Each entry of the table gives the probability of obtaining as output the 
joint outcome indexing its column given that the input was the
measurement context indexing its row.
This behaviour is formalized as an empirical model $\family{e_\sigma}_{\sigma \in \Sigma}$, with the entries in each row of the table directly specifying the probability distribution for a facet of $\Sigma$.
It is straightforward to check that these distributions are compatible.
The probability distributions for the non-maximal faces can then be obtained by marginalization.
Note that these marginals are well-defined if and only if compatibility holds.

\begin{table}[tbp]
\centering
$\XYTablePR$
\caption{A PR box.}
\label{tab:example}
\end{table}

In this example, compatibility ensures that the local behaviour on Alice's part of the box, as described by the probability distributions $e_{\enset{x_1}}$ and $e_{\enset{x_2}}$, is independent of Bob's choice of input, and vice versa.
If this were not the case, then it would be possible e.g. for Bob to use the box to instantaneously signal to Alice by altering her locally observable behaviour through his choice of input.
\end{example}

\begin{definition}
An empirical model $e : \XSO$ is said to be \emph{non-contextual} if it is compatible with a global section for ${\Dist \circ \Ev}$. In other words, $e$ is non-contextual if there exists some $d \in \Dist \circ \Ev(X)$, a distribution over global assignments of outcomes to measurements, such that $d |_\sigma = e_\sigma$ for all measurement contexts $\sigma \in \Sigma$.
Otherwise, the empirical model is said to be \emph{contextual}.
\end{definition}

Noncontextuality characterizes classical behaviours.
One way to understand this is that it reflects a situation in which the physical system being measured exists at all times in a definite state assigning outcome values to all properties that can be measured.
Probabilistic behaviour may still arise, but only via stochastic mixtures or distributions on  these global assignments.
This may reflect an averaged or aggregate behaviour,
or an epistemic limitation on our knowledge of the underlying global assignment.

\section{The algebraic viewpoint}\label{sec:algebraic}

\subsection{Operations on empirical models}\label{ssec:operations}

We consider operations that transform and combine emprical models to form new ones.
One should think of these as elementary operations that can be carried out
classically,
\ie without using contextual resources beyond the empirical models given as arguments.
For this reason, these operations are regarded as `free' in the resource theory of contextuality.

Most of the operations presented here
were introduced by a subset of the authors in \cite{abramsky2017contextual}.
A novelty is the idea of conditional measurement,
which is intended to capture (a one-step version of)
the kind of classical control of quantum systems that is used in 
adaptive measurement-based quantum computation schemes. Iterating this construction yields longer protocols of this kind.

For each operation, we give some brief motivating explanation followed by its definition.
All the operations are summarized in Table~\ref{table:ops}, as typing rules.

\newcommand{\eone}{\mathsf{u}}
\newcommand{\ezero}{\mathsf{z}}

\begin{itemize} 
\itemd{Zero model}
Consider the unique scenario with no measurements:
\[\tuple{\emptyset,\Delta_0 = \enset{\emptyset},\family{}} \Mdot\]
There is a single empirical model on this scenario, which we denote by $\ezero$.

\itemd{Singleton model} 
Consider the unique scenario that has a single measurement with a single outcome: 
\[\tuple{\one = \enset{\star},\Delta_1 = \enset{\emptyset, \one},\family{O_\star = \one}} \Mdot\]
There is a single empirical model on this scenario, which we denote by $\eone$.

\itemd{Translation of measurements}
From an empirical model in a given scenario,
we can build another in a different scenario,
by mapping the measurements in the latter scenario to those in the former, taking care to respect compatibilities.
In particular, this can capture the operation of restricting the allowed measurements (or the compatibilities).
Note that it can also mean that two measurements in the new scenario are just different aliases for the same measurement being performed in the original model.

The preservation of compatibilities is captured by the notion of simplicial map.
Given simplicial complexes $\Sigma$ and $\Sigma'$ on sets of vertices $X$ and $X'$, respectively,
a simplicial map $\fdec{f}{\Sigma}{\Sigma'}$ is a function between the vertex sets, $\fdec{f}{X}{X'}$,
that maps faces of $\Sigma$ to faces of $\Sigma'$, \ie such that for all $\sigma \in \Sigma$, $f(\sigma) \in \Sigma'$.

Given an empirical model $e : \XSO$ and a simplicial map 
 $\fdec{f}{\Sigma'}{\Sigma}$, 
the model $f^*e : \XpSpfO$,
where $(f^*O)_x\defeq O_{f(x)}$ for all $x \in X'$,
is defined by pulling $e$ back along the map $f$:
for any $\sigma \in \Sigma'$ and $s \in \Ev_{f^*O}(\sigma)$,
\[(f^*e)_\sigma(s) \defeq \sum_{\substack{ t\in \Ev_O(f(\sigma)) \\ t \circ f|_{\sigma} = s}}e_{f(\sigma)}(t) \Mdot\]

Concretely, $f^*e$ can be implemented from $e$ as follows: when a measurement $x\in X'$ is to be performed, one performs $f(x)$ instead. Requiring $f$ to be a simplicial map guarantees that any set of compatible measurements in $\Sigma'$ can indeed be jointly measured in this manner.

\itemd{Coarse-graining of outcomes}
We can similarly consider a translation of outcomes.

Given $e : \XSO$ and a family of functions $h={\family{\fdec{h_x}{O_x}{O'_x}}_{x \in X}}$,
 define an empirical model $e/h$ on the scenario $\XSOp$
as follows: for each $\sigma \in \Sigma$ and $s \in \Ev_{O'}(\sigma)$ 
\[(e/h)_\sigma(s) \defeq \sum_{\substack{t \in \Ev_O(\sigma) \\ \left(\prod_{x\in\sigma}h_x\right) \circ t = s}}e_\sigma(t)\Mdot\]

One can use $e$ to implement $e/h$: one just performs the measurement and applies the corresponding function $h_x$ to the outcome obtained. 

\itemd{Probabilistic mixing}
We can consider convex combinations of empirical models:
from two models on the same scenario,
a new model is constructed in which a coin, which may be biased, is flipped to choose which of the two models to use.

Given empirical models $e$ and $e'$ in $\XSO$ and a probability value $\lambda \in [0,1]$,
the model $e +_\lambda e' : \XSO$ is given as follows:
for any $\sigma \in \Sigma$ and $s \in \Ev(\sigma)$,
\[(e +_\lambda e')_\sigma(s) \defeq \lambda \, e_\sigma(s) \; + \; (1-\lambda) \, e'_\sigma(s) \Mdot\]

\itemd{Controlled choice}
From two empirical models, we can construct a new model that can be used as either one or the other.
The choice is determined by which measurements are performed,
but the compatibility structure enforces that only one of the original models ends up being used.

Let $e : \XSO$ and $e' : \XpSpOp$ be empirical models.
We consider a new scenario built out of these two.
The measurements are ${X \sqcup X' \defeq \enset{0} \times X \cup  \enset{1} \times X'}$,
with outcomes given accordingly by the copairing $[O,O']$, i.e.:
\begin{align*}
[O, O']_{(0, x)} &= O_x \quad\text{for $x \in X$} \Mcomma
\\
[O, O']_{(1, x)} &= O'_x \quad\text{for $x \in X'$} \Mdot
\end{align*}
The contexts are given by the coproduct of simplicial complexes
\[\Sigma + \Sigma' \defeq  \setdef{\enset{0}\times\sigma}{\sigma \in \Sigma} \cup \setdef{\enset{1}\times\sigma}{\sigma \in \Sigma} \Mcomma\]
ensuring that all the measurements performed in a single use come from the same of the two original scenarios, so that only one of the empirical models is used. 

The empirical model $e \choice e' : \tuple{X \sqcup X', \Sigma + \Sigma',[O, O']}$
is given as
\begin{align*}
(e \choice e')_{\enset{0} \times \sigma} &\defeq e_\sigma \quad\text{for $\sigma \in \Sigma$} \Mcomma
\\
(e \choice e')_{\enset{1} \times \sigma} &\defeq e'_\sigma \quad\text{for $\sigma \in \Sigma'$} \Mdot
\end{align*}

\itemd{Tensor product}
From two empirical models, a new one is built that
allows the use of both models \textit{independently}, in parallel.

Let $e : \XSO$ and $e' : \XpSpOp$ be empirical models.
As in the previous case,
consider a new scenario with measurements $X \sqcup X'$ and outcomes $[O,O']$.
The difference is that the contexts are now given by the simplicial join
\[\Sigma \star \Sigma' \defeq \setdef{\sigma \sqcup \sigma'}{\sigma \in \Sigma, \sigma' \in \Sigma'} \Mcomma\]
corresponding to the fact one may use measurements from the two scenarios in parallel.
The empirical model 
$e \otimes e' : \tuple{X \sqcup X', \Sigma \star \Sigma',[O , O']}$,
is given as
\[(e \otimes e')_{\sigma \sqcup \sigma'}\tuple{s,s'} \defeq e_C(s) \, e'_{C'}(s')\]
for all $\sigma \in \Sigma$,  $\sigma' \in \Sigma'$, $s \in \Ev_{O}(\sigma)$, and $s' \in \Ev_{O'}(\sigma')$.

\itemd{Conditioning on a measurement}
Given an empirical model, one may perform two compatible measurements in sequence.
But in such a situation,
when one decides to perform the second measurement, the outcome of the first is already known.
We could therefore consider the possibility of choosing which measurement to perform second
depending on the outcome observed for the first measurement.
This process can be considered as a 
measurement in its own right yielding as its outcome the pair consisting of the outcomes of the two constituent
measurements.
We can extend the original empirical model with such an extra measurement. 

In order to define this operation, we need the concept of link of a face in a simplicial complex.
Given a simplicial complex $\Sigma$ and a face $\sigma \in \Sigma$,
the link of $\sigma$ in $\Sigma$ is the subcomplex of $\Sigma$ whose faces are
\[\lk_\sigma\Sigma \defeq \setdef{\tau \in \Sigma}{\sigma \cap \tau = \emptyset, \sigma \cup \tau \in \Sigma} \Mdot\] 
If we think of $\Sigma$ as representing the compatibility structure of a measurement scenario,
and suppose that the measurements in a face $\sigma$ have already been performed,
then the complex $\lk_\sigma\Sigma$
represents the compatibility structure of the measurements that may still be performed,
ensuring that overall one always performs a set of compatible measurements according to $\Sigma$.

Let $e : \XSO$ be an empirical model,
and take a measurement $x \in X$ and a family of measurements
$\family{y_o}_{o \in O_x}$ with $y_o \in \Vertices(\lk_{\enset{x}}\Sigma)$
a vertex of the complex
\[\lk_{\enset{x}}(\Sigma) = \setdef{\sigma}{x \not\in \sigma, \enset{x} \cup \sigma \in \Sigma}\Mdot\]
Consider a new measurement $x ? \family{y_o}_{o \in O_x}$, abbreviated $x ? y$.
We call such a measurement a \emph{conditional measurement}.
Its outcome set is the dependent pair type
\[{O}_{x ? y} \defeq \bigsqcup_{o\in O_x} O_{y_o} = \setdef{(o,o')}{o \in O_x, o' \in O_{y_o}} \Mdot\]
The compatibility structure is extended to take the new measurement into account:
\[\Sigma[x?y] \defeq \Sigma \cup \setdef{\sigma \cup \enset{x?y}}{\Forall{o \in O_x} \sigma \cup \enset{x,y_o} \in \Sigma}\Mdot\]

Define the new model
\[e[x?y] : \tuple{X \cup \enset{x?y},\Sigma[x?y],O[x?y\mapsto O_{x?y}]}\]
as follows: for the old faces $\sigma \in \Sigma$,
\[
e[x?y]_\sigma \defeq e_\sigma
\Msemicolon
\]
for the new faces of the form $\sigma \cup \enset{x?y}$ satisfying ${\sigma \cup \enset{x,y_o} \in \Sigma}$ for all $o \in O_x$, we have,
for any $s \in \Ev(\sigma)$ and $(o,o') \in O_{x?y}$,
\begin{equation*}
\begin{multlined}
e[x?y]_{\sigma \cup \enset{x?y}}(s[x?y \mapsto (o,o')]) \defeq \\
\begin{cases}
   \mathrlap{e_{\sigma \cup \enset{x,y_o}}(s[x\mapsto o, y_o \mapsto o'])} & \\
      & \text{ if $x \in \sigma \implies s(x)=o$ and  $y_o \in \sigma \implies s(y_o)=o'$} \\
   0  & \text{ otherwise.}
\end{cases}
\end{multlined}
\end{equation*}
\end{itemize}

\subsection{The contextual fraction}\label{ssec:cf}

The contextual fraction is a quantitative measure of the degree to which any empirical model exhibits contextuality \cite{ab,abramsky2017contextual}, which we define here using the operation of probabilistic mixing.

\begin{definition}
Given an empirical model $e : \XSO$, we consider the set of all possible decompositions
\[ e = e^{\mathsf{NC}} +_\lambda e' \Mcomma \]
such that $e^{\mathsf{NC}}$ is non-contextual.
The \emph{non-contextual fraction} of $e$, denoted $\NCF(e)$, is defined to be the maximum value of $\lambda$ over all such decompositions.
The \emph{contextual fraction} of $e$, denoted $\CF(e)$, is then defined as
\[ \CF(e) \defeq 1 - \NCF(e) \Mdot \]
\end{definition}

A crucial property for a useful measure of contextuality is that it should be a monotone for the free operations of our resource theory.
That is, it should be non-increasing under those elementary operations on empirical models that can be carried out classically.

\begin{proposition}\label{prop:ncfandoperations}
For the operations we have introduced in Section~\ref{ssec:operations} the contextual fraction satisfies the following monotonicity properties.
\begin{itemize}
\item $\CF(\ezero) = \CF(\eone) = 0$
\item $\CF(f^*e) \leq \CF(e)$
\item $\CF(e/h) \leq \CF(e)$
\item $\CF(e +_\lambda e') \leq \lambda \CF(e) + (1-\lambda) \CF(e')$
\item $\CF(e \choice e') = \max\enset{\CF(e),\CF(e')}$
\item $\CF(e \otimes e') = \CF(e) + \CF(e') - \CF(e)\CF(e')$, \\
\ie $\NCF(e \otimes e') = \NCF(e)\NCF(e')$
\item $\CF(e[x?y]) = \CF(e)$
\end{itemize}
\end{proposition}

\begin{proof}
We will only show that $\CF(e[x?y]) = \CF(e)$, as the other statements were proved in~\cite{abramsky2017contextual}. The inequality ${\CF(e) \leq \CF(e[x?y])}$ holds by monotonicity of measurement translation, since $e = f^* e[x?y]$ where ${\fdec{f}{\Sigma}{\Sigma[x?y]}}$ is the inclusion map.

For the other direction, note that $\text{--}[x?y]$ preserves convex combinations and deterministic empirical models,
\ie those models that arise from (a delta distribution on) a single global assignment. Since non-contextual models are precisely those that are convex combinations of deterministic models, the operation $\text{--}[x?y]$ takes non-contextual models to non-contextual models. Now, if $e=e^{\mathsf{NC}} +_\lambda e'$ where $e^{\mathsf{NC}}$ is non-contextual, then ${e[x?y]=e^{\mathsf{NC}}[x?y] +_\lambda e' [x?y]}$ and  $e^{\mathsf{NC}}[x?y]$ is also non-contextual. Consequently, $\NCF(e[x?y])\geq \NCF(e)$, which means that $\CF(e[x?y])\leq\CF(e)$.
\end{proof}

\subsection{Equational theory}\label{ssec:eqtheory}
We consider terms built out of variables and the operations of Section~\ref{sec:algebraic}:
\begin{align*}
\mathsf{Terms} \ni t \;\;\coloncolonequals\;\; &
  v \in \mathsf{Var}
\;\mid\;  \ezero
\;\mid\;  \eone
\;\mid\;  f^* t 
\;\mid\;  t/h \\
\;\mid&\; t +_\lambda t
\;\mid\;  t \choice t
\;\mid\;  t \otimes t
\;\mid\;  t[x?y]
\end{align*}
according to the `typing' rules in Table~\ref{table:ops}.
Note that, due to the restriction forbidding repeated variables when building  contexts,
there is at most one occurrence of each variable in each well-typed term. 
We can interpret such a term as a composed `free' operation on empirical models.
More specifically, the typed term
\[v_1 : \X_1, 
\ldots, v_n : \X_n 
\vdash t : \X 
\]
represents an operation that takes $n$ empirical models,
$e_i : \X_i$ 
for $i \in \enset{1,\ldots,n}$,
and builds a new empirical model, denoted $t[e_1/v_1,\ldots,e_n/v_n]$,
on the scenario $\X$, 
according to the definition of the elementary operations given in the itemized list above.

Terms without variables should therefore correspond to empirical models that are `free' as resources. Indeed, they are precisely the non-contextual ones.

\begin{proposition}\label{prop:novarterm}
 A term without variables always represents a non-contextual empirical model. Conversely, every non-contextual empirical model can be represented by a term without variables.
\end{proposition}
\begin{proof}
  Using Proposition~\ref{prop:ncfandoperations}, it is straightforward to show by induction that every term $t$ without variables satisfies $\CF(t)=0$, which proves the first claim. 
  
  For the second claim, note that since probabilistic mixing is an allowed operation and non-contextual empirical models are precisely the mixtures of deterministic models, it suffices to show that every deterministic empirical model can be built from the operations. So let $e \colon \XSO$ be deterministic,
  and write $e(x)\in O_x$ for the certain outcome it assigns to measurement $x \in X$.
  Then $e=(f^*\eone)/h$,  where $f$ is the unique simplicial map $\Sigma\lto\enset{\star}$ and $h=\family{\fdec{h_x}{\enset{\star}}{O_x}}_{x\in X}$ is defined by $h_x(\star)=e(x)$.
\end{proof}

\begin{table*}[t]%
\caption{Free operations on empirical models}\label{table:ops}
\centering
{\begin{tabular}{|l l|}\hline & \\
\multicolumn{2}{|c|}{$v \in \mathsf{Var}$ \quad\quad\quad $t, t' \in \mathsf{Terms}$ \quad\quad\quad
$\XSO \in \mathsf{Scenarios}$ \quad\quad\quad $\Gamma, \Gamma' \coloncolonequals \emptyset \mid \enset{v : \XSO}, \Gamma$ ($v \not\in\Gamma$)}
\\ & \\
\multicolumn{2}{|l|}{{\rvar} \qquad\qquad\qquad\quad {\rzero} \qquad\qquad\qquad\quad {\rsingl}}
\\ &\\
{\rmeas} & {\routc}
\\ &\\
{\rmix} & {\rchoice} 
\\ &\\
{\rcond} & {\rprod}
\\ & \\ \hline
\end{tabular}}
\end{table*}

We present a list of equations \eqref{eq:thfirst}--\eqref{eq:thlast} between terms, with variables denoted $a,b,c,d$. These equations are supposed to capture equality up to a static congruence -- isomorphism of empirical models defined below. Implicit is the assumption that the terms on both sides of the equality signs are well-typed in the same typing context, according to the rules of Table~\ref{table:ops}.
That is, as we shall see in Proposition~\ref{prop:soundness}, when we write $t = t'$,
we are implicitly thinking of all the typing contexts 
$\Gamma$ such that
$\Gamma \vdash t : \X$ and $\Gamma \vdash t' :\X'$.


For most of these equations,
it is enough that the term on the left-hand side be well-typed in a given context for the term on the right-hand side
to also be.
The exception to this rule is equation \eqref{eq:commentaboutdefined}.
It is important not to be misled into reading it as meaning that any two consecutive extensions with conditional measurements commute.
This is only the case when both are conditional measurements of the original model,
\ie when the second conditional measurement does not make use of the first.
In the notation of the equation in question, assuming that the term on the left is well-typed,
we would require that $x' \neq x?y$ and $y'_o \neq x?y$ for all $o \in O_x$ in order to be able to also type the term on the right.

\begin{itemize}
\item The controlled choice is a commutative monoid with neutral element $\ezero$:
\begin{shiftflalign}&
\label{eq:thfirst}
a \choice b = b \choice a 
&\\&
a \choice (b \choice c) = (a \choice b) \choice c 
&\\&
a \choice \ezero = a = \ezero \choice a
&\end{shiftflalign}
\item The product is a commutative monoid with neutral element $\ezero$:
\begin{shiftflalign}&
a \otimes b = b \otimes a
&\\&
a \otimes (b \otimes c) = (a \otimes b) \otimes c
&\\&
a \otimes \ezero = a = \ezero \otimes a
&\end{shiftflalign}
\item Standard axioms of convex combinations:
\begin{shiftflalign}&
a +_0 b = a
&\\&
a +_\lambda a = a
&\\&
a +_\lambda b = b +_{1-\lambda} a \label{eq:convexcomm}
&\\&
(a +_\lambda b) +_{\lambda'} c = a +_{\lambda\lambda'} (b +_{\frac{\lambda'-\lambda\lambda'}{1-\lambda\lambda'}} c) \label{eq:convexassoc}
&\end{shiftflalign}

\item Measurement and outcome transformations:
\begin{shiftflalign}&
g^*(f^*a) = (f \circ g)^* a \label{eq:transmeas}
&\\&
(a/h)/j = a/(j \circ h) \label{eq:transoutc}
&\\&
f^*(a/h) = f^*a / f^*h \label{eq:transmeasoutc}
&\end{shiftflalign}
where $\fdec{(f^*h)_x = h_{f(x)}}{O_{f(x)}}{O'_{f(x)}}$.

\item Convex combinations and the other operations:
\begin{shiftflalign}&
f^*(a +_\lambda b) = f^*a +_\lambda f^*b \label{eq:convexotherfirst}
&\\&
(a +_\lambda b)/h = a/h +_\lambda b/h
&\\&
(a +_\lambda b) \choice (c +_\lambda d) = (a \choice c) +_\lambda (b \choice d) 
&\\&
(a +_\lambda b) \otimes c = (a \otimes c) +_\lambda (b \otimes c) 
&\\&
(a+_\lambda b)[x?y] = a[x?y] +_\lambda b[x?y] \label{eq:convexotherlast}
\end{shiftflalign}
\item  Transformations and other operations:
\begin{shiftflalign}&
a/h \choice b/j = (a \choice b)/[h,j] \label{eq:transotherfirst}
&\\&
a/h \otimes b/j = (a \otimes b)/[h,j]\label{eq:transothertensorout}
&\\&
f^*a \choice g^*b = [f,g]^* (a \choice b)
&\\&
f^*a \otimes g^*b = [f,g]^* (a \otimes b) \label{eq:transotherlasst}
&\end{shiftflalign}
\item Conditional measurements and other operations
\begin{shiftflalign}&
a[x?y][x'?y'] = a[x'?y'][x?y] \label{eq:commentaboutdefined} 
&\\&
(f^*a)[x?y] = \tilde{f}^*(a[f(x)?(f\circ y)]) \label{eq:condmeas}
&\end{shiftflalign}
where, for $\fdec{f}{\Sigma'}{\Sigma}$, we have that  $\fdec{\tilde{f}}{\Sigma'[x?y]}{\Sigma}$ is the extension
$\tilde{f} \defeq f[x?y \mapsto f(x)?(f\circ y)]$.
\begin{shiftflalign}&
(a/h)[x?y]=(a[x?(y\circ h_x)])/\tilde{h} \label{eq:condoutc}
&\end{shiftflalign}
where for $\family{\fdec{h_x}{O_x}{O'_x}}_{x \in X}$,
we have $(y\circ h_x)_o = y_{h_x(o)}$ for all $o \in O_x$, and
$\tilde{h}$ extends the family $h$ with $\tilde{h}_{x?(y\circ h_x)}$ mapping a pair $(o\in O_x, \, o'\in O_{(y\circ h_x)_o})$ to $(h_x(o)\in O'_x, \, h_{(y\circ h_x)_o}(o') \in O'_{y_{h_x(o)}})$.
\begin{shiftflalign}&
a[x?y] \choice b = (a \choice b)[x?y] \label{eq:condchoice}
&\\&
a[x?y] \otimes b = (a \otimes b)[x?y] \label{eq:condtimes}
&\end{shiftflalign}
\item Choice can be eliminated:
\begin{shiftflalign}&
a \choice b = i^*(a \otimes b)\label{eq:thlast}
&\end{shiftflalign}
where $\fdec{i}{\Sigma + \Sigma'}{\Sigma\star\Sigma'}$ is the inclusion of simplicial complexes (it acts as identity on the vertices).
\end{itemize}

The above equational theory intends to capture equality up to the following notion of isomorphism.

\begin{definition}
 Two empirical models $e: \XSO$ and $d: \XpSpOp$ are said to be isomorphic, written $e \simeq d$, if there is a simplicial isomorphism $\isofdec{f}{\Sigma'}{\Sigma}$ and a family of bijections $\family{\isofdec{h_x}{O_{f(x)}}{O'_x}}_{x\in X'}$ such that $d = (f^*e)/h$. 
\end{definition}

Note that these isomorphisms coincide exactly with the isomorphisms of  the category \cat{Emp} that is defined in the next section.

\begin{proposition}[Soundness]\label{prop:soundness}
The equational theory given by equations \eqref{eq:thfirst}--\eqref{eq:thlast} is sound.
That is, if $t=t'$ is one of these equations, then for any context $\Gamma = \enset{v_1 : \X_1, \ldots, v_n : \X_n}$
such that $\Gamma \vdash t : \X$ and $\Gamma \vdash t' : \X'$
and for any empirical models $e_1 : \X_1, \ldots, e_n : \X_n$,
we have
\[t[e_1/v_1, \ldots, e_n/v_n] \simeq t'[e_1/v_1, \ldots, e_n/v_n] \Mdot\]
\end{proposition}

The proof is a tedious but straightforward verification of the conditions.
It is an open question whether this equational theory is complete.
An important step towards proving completeness -- or towards finding the missing equations --
is provided by the following normal form result.
It establishes that, using the equations,
we can transform any term into one where the operations are applied in a certain order.

\begin{proposition}[Normal form]\label{prop:normalform}
Let $\Gamma \vdash t : \X$. 
Then $t$ can be rewritten using equations \eqref{eq:thfirst}--\eqref{eq:thlast}
into a term $t_0$ of the following form:  
\begin{align*}
    t_0 &\coloncolonequals t_1 \mid t_0 +_\lambda t_1
\\
    t_1 &\coloncolonequals (f^*t_2)/h
\\
    t_2 &\coloncolonequals t_3 \mid t_2[x?y]
\\
    t_3 &\coloncolonequals t_4 \mid t_4 \otimes t_3
\\
    t_4 &\coloncolonequals \ezero \mid \eone \mid v \in \textsf{Var}
\end{align*}
\end{proposition}
\begin{proof}
We are always using the rules from left to right (see remark immediately before the equations).

First, note that rule \eqref{eq:thlast} allows us to rewrite the choice operation in terms of the others, so we can assume that $t$ has no occurrences of choice.

Using rules \eqref{eq:convexotherfirst}--\eqref{eq:convexotherlast},
all uses of probabilistic mixing can be taken to the top level.
By the kind of associativity rule \eqref{eq:convexassoc},
 $t$ can then be rewritten to a term of the form
$t_0 = t^{1} +_{\lambda_1} (t^2 +_{\lambda_2} (\cdots t^n))$ where the terms $t^{i}$ do not use the mixing operation.

Now, let $t_1$ be a term without mixing. By a similar argument,
using equations \eqref{eq:transothertensorout}, \eqref{eq:transotherlasst} and \eqref{eq:condmeas}--\eqref{eq:condoutc}, we can commute translation of measurements and outcomes to the top level relative to the remaining operations.
Using \eqref{eq:transmeasoutc}, all coarse-grainings of outcomes can be commuted upwards, and using \eqref{eq:transmeas} and \eqref{eq:transoutc}, one can combine consecutive applications of either of these two operations. Therefore $t_1$ can be rewritten as $(f^*t_2)/h$ for some $f$ and $h$ and some term $t_2$ without occurrences of mixing, translation of measurements, or coarse-graining of outcomes.

From rule \eqref{eq:condtimes},
$t_2$ can be rewritten to have the form $t_3[x_1?y_1]\cdots[x_n?y_n]$ where $t_3$ only uses base cases and the product operation.
\end{proof}

\section{The categorical viewpoint}\label{sec:categorical}

\newcommand{\ntsigma}{h}
\newcommand{\nttau}{j}

In this section we make precise the idea of using one empirical model to simulate the behavior of another one. In fact, there are several notions of a simulation, depending on the powers allowed to those doing the simulating. The simplest notion is deterministic and has a clear intuitive meaning: to use $d\colon\Y $ to simulate $e\colon\X$ amounts to giving a measurement $\pi(x)\in Y$ for every $x\in X$ and a deterministic way of interpreting the outcomes of $\pi(x)$ as outcomes of $x$, \ie a map $P_{\pi(x)} \lto O_x$. For such a protocol to succeed, $\pi$ has to be simplicial and the outcome statistics of $d$ must, after interpretation, give rise to the statistics of $e$.

One can then consider ways of extending such simulations. In~\cite{karvonen2018categories} more expressive power was obtained by  allowing $\pi(x)$ to be a set of measurements instead of a singleton, and by allowing outcomes to be interpreted stochastically. 

Here we obtain even more general simulations by letting $\pi(x)$ be an adaptive measurement protocol on $\Y$. Classical shared randomness can then be modelled by allowing the use of auxilliary non-contextual empirical models.

\subsection{Deterministic simulations}

\begin{definition}
    Let $\X=\XSO$ and $\Y = \YTP$ be measurement scenarios. A \emph{deterministic morphism} $\fdec{\tuple{\pi,\ntsigma}}{\Y}{\X}$ consists of: 
   \begin{itemize}
        \item a simplicial map $\fdec{\pi}{\Sigma}{\Theta}$;
        \item a natural transformation $\fdec{\ntsigma}{\Ev_P\circ\pi}{\Ev_O}$;
        equivalently, 
         a family of maps $\fdec{\ntsigma_x}{P_{\pi(x)}}{O_x}$ for each $x\in X$.
    \end{itemize}
    The composite of the morphisms $\fdec{\tuple{\rho,\family{\nttau_y}_{y\in Y}}}{\Z}{\Y}$ and $\fdec{\tuple{\pi,\family{\ntsigma_x}_{x\in x}}}{\Y}{\X}$ is given by $\tuple{\rho\circ\pi,\family{\ntsigma_x\circ\nttau_{\pi (x)}}_{x\in X}}$. 
   
    Given an empirical model $d\colon \Y$, its pushforward along a deterministic morphism $\tuple{\pi,\ntsigma}$ is the empirical model $\tuple{\pi,\ntsigma}_* d\colon\X$
    defined by: for any $\sigma \in \Sigma$,
            \[(\tuple{\pi,\ntsigma}_* d)_\sigma=\Dist(\ntsigma_\sigma)(d|_{\pi(\sigma)}) \Mdot\]
    Let $e\colon \X$ and $d\colon \Y$ be empirical models. Then a \emph{deterministic simulation} $\fdec{\tuple{\pi,\ntsigma}}{d}{e}$ is a deterministic morphism  $\fdec{\tuple{\pi,\ntsigma}}{\Y}{\X}$ such that 
            \[ e=\tuple{\pi,\ntsigma}_* d \Mdot\]
 The category of 
 empirical models and deterministic simulations is denoted by \cat{Emp}.
\end{definition}

The reason that natural transformations $\fdec{\ntsigma}{\Ev_P\circ\pi}{\Ev_O}$ correspond to families of maps $\fdec{\ntsigma_x}{P_{\pi(x)}}{O_x}$ for each $x\in X$ is the following: both $\Ev_P\circ\pi$ and $\Ev_O$ are sheaves on a discrete space, and hence morphisms can be glued along any covering -- and in particular along the covering by singletons.

The category \cat{Emp} and the category of measurement scenarios are in fact symmetric monoidal with the product defined in Section~\ref{ssec:operations}. The action on morphisms is given by
\[\tuple{\pi,(\ntsigma_x)_{x\in X}}\otimes \tuple{\pi',(\ntsigma'_y)_{y\in Y}}=\tuple{\pi\sqcup \pi',(\ntsigma_x)_{x\in X}\sqcup (\ntsigma'_y)_{y\in Y}}\]

\subsection{Measurement protocols}\label{ssec:mps}

Deterministic simulations are fairly limited in their expressive power. For instance, one might want to use classical randomness in simulations. If one thinks of an empirical model as a black box, even more is possible: one could first perform a measurement, then based on the observed outcome choose which compatible measurement to perform next, and so on. Such procedures are known as measurement protocols \cite{acin2015combinatorial} or wirings \cite{allcock2009closedsets} in the literature on non-locality.

The main task of this section is to formalize carefully the notion of measurement protocol. We define an operation that takes a measurement scenario $\X$ and builds a new scenario $\MP(\X)$, whose measurements are all the measurement protocols over $\X$.
This turns out to be functorial, and moreover, a comonad.
Hence, using the co-Kleisli category, one can think of more general simulations $d \lto e$ as deterministic simulations $\MP(d)\lto e$.
Moreover, we will see that classical randomness comes for free by considering simulations of the form $\MP(d\otimes c)\lto e$, where $c$ is a non-contextual empirical model.

A measurement protocol is a certain kind of decision tree: the root is the first measurement, and the outcomes obtained dictate which measurements to choose next, \ie which branch of the tree to pick. Rooted trees can be formalized in various ways:
\eg recursively, as certain graphs, or in terms of prefix-closed sets of words.
We have chosen to formalize them using the latter approach.
We try to keep the intuitive picture in mind as it can give sense to the proofs, which may seem somewhat technical otherwise.

\begin{definition}
    A \emph{run} on a measurement scenario $\X = \XSO$ is a sequence $\x\defeq(x_i,o_i)_{i=1}^n$ such that $x_i\in X$ are distinct, $\enset{x_1,\ldots,x_n}\in\Sigma$, and each $o_i\in O_{x_i}$.
\end{definition}
    A run $\x$ determines a context
    $\sigma_{\x}\defeq\enset{x_1,\ldots, x_n} \in \Sigma$ and a joint assignment on that context $\fdef{s_{\x}}{x_i}{o_i} \in \Ev(\sigma_{\x})$. 
    Two runs $\x$ and $\y$ are said to be \emph{consistent} if they agree on common measurements, \ie  for every $z\in \sigma_{\x}\cap \sigma_{\y}$ we have $s_{\x}(z)=s_{\y}(z)$.
    
Given runs $\x$ and $\y$, we denote their concatenation by $\x\cdot\y$. Note that $\x\cdot\y$ might not be a run.

\begin{definition}
    A \emph{measurement protocol} on $\X = \XSO$ is a non-empty set $Q$ of runs satisfying the following  conditions:
    \begin{enumerate}[(i)]
        \item if $\x\cdot \y\in Q$ then $\x\in Q$;
        \item\label{cond:allpossibleoutcomes} if $\x\cdot (x,o)\in Q$, then $\x\cdot (x,o')\in Q$ for every $o'\in O_x$; 
        \item\label{cond:determinacy} if $\x\cdot (x,o)\in Q$ and $\x\cdot (x',o')\in Q$, then $x=x'$. 
    \end{enumerate}
\end{definition}
    One can think of such a measurement protocol as a (deterministic) strategy for interacting with an empirical model, seen as a black box whose interface is given by its measurement scenario: Condition~\eqref{cond:determinacy} expresses that the previously observed outcomes determine the next measurement to be performed, while Condition~\eqref{cond:allpossibleoutcomes} captures the fact that every outcome of a performed measurement may in principle be observed and so the protocol must specify how to react to each possibility, either by performing a new measurement or by stopping.

\begin{definition}
Given a scenario $\X = \XSO$, we build a scenario $\MP(\X)$:
   \begin{itemize}
       \item its set of measurements is the set $\MP(X)$ of measurement protocols on $\X$;
       \item the outcome set $O_Q$ of a measurement protocol $Q \in \MP(X)$ is its set of maximal runs, \ie those $\x\in Q$ that are not a proper prefix of any other $\y\in Q$;
       \item a set $\enset{Q_1,\ldots Q_n}$ of measurement protocols is compatible
       whenever for any choice of pairwise consistent runs
       $\x_i \in Q_i$ with $i \in \enset{1,\ldots,n}$,
       we have $\bigcup_i \sigma_{\x_i} \in \Sigma$.
   \end{itemize}
  \end{definition}
  %
  %

\begin{definition}
    Given an empirical model $e\colon \X$, we define the empirical model $\MP(e) \colon \MP(\X)$ as follows.
    For a compatible set $\sigma\defeq\enset{Q_1,\ldots, Q_n}$ of measurement protocols and an assignment $\fdef{s}{Q_i}{\x_i} \in \Ev(\sigma)$, we set
\[\MP(e)_\sigma(s)\defeq\begin{cases} e_{\bigcup_{i} \sigma_{\x_i}}(\cup_i s_{\x_i}) &\text{ if $\enset{\x_i}$ pairwise consistent}\\
    								0 &\text{ otherwise. }\end{cases}\]
\end{definition}

One way of thinking about the definition above is that it identifies a measurement protocol with the set of all situations in which one might find oneself while carrying out the protocol. Informally, compatibility of measurement protocols means they can be interleaved in any order whatsoever,
and when running them one never ends up performing incompatible measurements.

It is clear that $\MP(e)$ satisfies the compatibility (or no-signalling) condition: calculating the probability of a joint outcome in $\MP(e)$ corresponds to a calculating a probability of a joint outcome in $e$, so no-signalling in $e$ gives no-signalling for $\MP(e)$. 

\subsection{Measurement protocols as a comonad}\label{ssec:comonad}

We now show that $\MP$ is in fact a comonad on empirical models. We do this by verifying the conditions for a  co-Kleisli triple.
We work on the category of measurement scenarios -- getting a comonad on empirical models requires little further effort.

First, we need to build a deterministic morphism 
\[\fdec{\epsilon_{\X}}{\MP(\X)}{\X}\Mdot\]
Intuitively, it is clear how this should be done: every measurement $x\in X$ can be viewed as a measurement protocol that only measures $x$ and stops afterwards. Formally, the simplicial map underlying $\epsilon_{\X}$ is defined by mapping each measurement $x\in X$ to the protocol
\[\enset{\Lambda} \cup \enset{x}\times O_{x} = \enset{\Lambda} \cup \setdef{(x,o)}{o \in O_x} \Mcomma\]
where $\Lambda$ is the empty word. The map of outcomes is given by sending an outcome (\ie maximal run) of this protocol $(x,o)$ to $o$.

Next, for scenarios $\X = \XSO$ and $\Y = \YTP$,
we define an extension operator that lifts a morphism \[\fdec{\tuple{\pi,\ntsigma}}{\MP(\X)}{\Y}\]
to a morphism 
\[\fdec{\tuple{\pi^\cok,\ntsigma^\cok}}{\MP(\X)}{\MP(\Y)}\Mdot\]
Given a measurement protocol $Q$ over $\Y$, we wish to define $\pi^\cok(Q)$. Intuitively, the extension works as follows: when running $Q$, any time one needs to perform a measurement $y$, one performs the measurement protocol $\pi(y)$ instead, maps its outcome to $P_y$ using $\ntsigma_y$, and consults the measurement protocol $Q$ to see what to do next. However, there is a slight catch: if a measurement protocol $\pi(y)$ requires one to perform a measurement in $\X$ that has already been done, there is no need to redo it -- one can simply reuse the previous outcome.
To make this precise, we first define inductively a merge operation $*$ for compatible runs:
	\begin{align*}
	    \x*\Lambda &\defeq x
	    \\
	    \x*((y,o)\cdot \y) &\defeq \begin{cases} \x*\y &\text{ if }y\in \sigma_{\x} \\
										(\x\cdot (y,o))*\y &\text{ otherwise.}
						   \end{cases}
						   	\end{align*}
We extend $*$ to all pairs of runs by setting $\x*\y=\Lambda$ whenever $\x$ and $\y$ are not compatible. 


We are now in a position to define $\pi^\cok Q$, but we first motivate the formal definition. What are the runs of $\pi^\cok Q$?
At least, they should contain those that can be interpreted as runs of $Q$: for instance, if $\y\defeq(y_i,p_i)_{i=1}^n\in Q$, then $\pi^\cok Q$ contains runs that can be interpreted as $\y$ by $\ntsigma^\cok$. These are precisely of the form $\x_1 * \dots * \x_n$ where each $\x_i\in\pi(y_i)$ is a maximal run and satisfies $\ntsigma_{y_i}(\x_i)=p_i$.
However, $\pi^\cok Q$ contains more runs, since any prefix of such an $\x_1 * \dots * \x_n$ must be included.

We define $\pi^\cok Q$ as the closure of the set 
\[\bigcup\setdef{\setdef{\x_1 * \dots * \x_n}{\x_i\in \ntsigma^{-1}_{y_i}(p_i)}}{(y_i,p_i)_{i=1}^n\in Q}\]
under taking prefixes. To see that $\pi^\cok Q$ defines a measurement protocol, note that (i) is automatically satisfied and (ii) and (iii) follow from the fact that each $\pi(y)$ is a measurement protocol. To define the map of outcomes $\ntsigma^\cok _Q$, note that the outcomes of $\pi^\cok Q$ are exactly words of the form $\x\defeq\x_1*\dots *\x_n$, where $\x_i\in\ntsigma^{-1}_{y_i}(p_i)$ with  $\y_{\x}\defeq(y_i,p_i)_{i=1}^n\in O_Q$. Since the word $\y_{\x}$ can be read from $\x$, we define $\ntsigma^\cok _Q$ by setting $\x\longmapsto \y_{\x}$. Altogether, we have an extension operator defined by \[\tuple{\pi,\family{\ntsigma_x}_{x\in X}}^\cok \defeq \tuple{\pi^\cok,\family{\ntsigma^\cok_Q}_{Q\in\MP(X)}}\Mdot\] 

To see why $\pi^\cok$ is simplicial, let $\enset{Q_1,\ldots, Q_n}$ be a compatible set of measurement protocols. Choose pairwise consistent runs $\x_i\in \pi^\cok Q_i$. By (ii) we may extend each $\x_i$ if necessary and assume that they are maximal. Now, we can define $\y_i\defeq\ntsigma^\cok_{Q_i}\x_i$ and consistency of $\x_i$ implies that $\y_i$ are consistent.
Because the protocols $Q_i$ are compatible, we have that $\cup_i \sigma_{\y_i}$ is a face. But then $\cup_i \sigma_{\x_i}=\cup_i \pi \left(\sigma_{\y_i}\right)=\pi \left(\cup_i \sigma_{\y_i}\right)$ is also a face since $\pi$ is simplicial. This shows that $\enset{\pi^\cok Q_1, \ldots, \pi^\cok Q_n}$ is compatible, as desired.


We now check that $\tuple{\MP,\epsilon,-^\cok}$ satisfies the axioms of a co-Kleisli triple, giving rise to a comonad on the category of measurement scenarios.
The equation \[\epsilon_{\X}^\cok=\id[\MP(\X)]\] is straightforward since the map on the left corresponds to a simulation doing the following: when the measurement protocol $Q$ calls for performing the measurement $x$, do the measurement protocol that consists only of $x$, and then proceed according to $Q$. This is just another way of describing the identity $\id[\MP(\X)]$. 

The axiom
\[\epsilon_{\Y}\circ \tuple{\pi^\cok,\ntsigma^\cok}=\tuple{\pi,\ntsigma}\]
is also easy: running the simulation on the left hand side corresponds to first thinking of a measurement $y\in Y$ as a measurement protocol that only does $y$, then applying $\pi$ to each measurement obtained on the way -- \ie to doing the simulation on the right hand side.

The last one requires that for $\fdec{\tuple{\rho,(\nttau_y)_{y\in Y}}}{\MP(\Z)}{\Y}$ and $\fdec{\tuple{\pi,\family{\ntsigma_x}_{x\in X}}}{\MP(\Y)}{\X}$, we have 
\[(\tuple{\pi,\ntsigma}\circ \tuple{\rho,\nttau}^\cok)^\cok=\tuple{\pi,\ntsigma}^\cok\circ \tuple{\rho,\nttau}^\cok\Mcomma\]
which can be rewritten as
\[
\tuple{\rho^\cok\circ \pi,\family{\ntsigma_{x}\circ\nttau^\cok_{\pi(x)}}_{x\in X}}^\cok 
=
\tuple{\rho^\cok\circ\pi^\cok,\family{\ntsigma^\cok_{Q}\circ \nttau^\cok_{\pi^\cok Q}}_{Q\in \MP(X)}}
\Mdot
\]
The left-hand side describes the following way of simulating a measurement protocol $Q\in\MP(X)$: whenever one needs to measure $x$, perform $\rho^\cok\pi(x)$ instead. The right-hand side describes the following way of simulating $Q$: whenever one needs to measure $x$, simulate it first by $\pi(x)$ and then simulate $\pi(x)$ by $\rho^\cok(\pi x)$, which amounts to the same thing. 

Making the proof of the first two equations more formal is not hard. Here, we do this for the third one.  To show that 
\[(\rho^\cok\circ \pi)^\cok=\rho^\cok\circ\pi^\cok \Mcomma\]
consider a  protocol
$Q$ over $\X$. Runs in $(\rho^\cok\circ \pi)^\cok Q$ are prefixes of runs that can be interpreted as runs $(x_i,o_i)_{i=1}^n\in Q$, \ie those of the form $(\z_1*\dots *\z_n)$ where $\ntsigma_{x_i}\nttau^\cok_{\pi(x)}\z_i=o_i$.
On the other hand, runs in $(\rho^\cok\circ\pi^\cok) Q$ are prefixes of runs that can be interpreted as runs of $\pi^\cok Q$.
In turn, runs of $\pi^\cok Q$ are prefixes of runs that can be interpreted as runs of $Q$:
they are prefixes of runs of the form $(\y_1*\dots *\y_n)$ where $\ntsigma_{x_i} y_i=o_i$ for some $(x_i,o_i)_{i=1}^n\in Q$. Hence, runs of $(\rho^\cok\circ\pi)^\cok Q$ are prefixes of runs of the form $(\z_1*\dots *\z_n)$ such that $\nttau^\cok_{\pi (x)}\z_i=\y_i$ for such $\y_i$.  But then each $\z_i$ also satisfies $\ntsigma_{x_i}\nttau^\cok_{\pi (x)}\z_i=o_i$, 
proving the desired equation.
Showing that the maps of outcomes agree is similar. 

We can therefore conclude that
$\MP$ defines a comonad on the category of measurement scenarios.
This comonad is in fact comonoidal: there is a canonical transformation
\[\MP(\text{--}\otimes \text{--})\lto \MP(\text{--})\otimes \MP(\text{--}) \Mcomma\] which corresponds to the fact that a measurement protocol over $\X$ can be seen as a measurement protocol over $\X\otimes \Y$ that never uses $\Y$. 

It is easy to check that the counit $\epsilon$, the co-Kleisli extension $\text{--}^\cok$, and the costrength lift to empirical models. That is,
\begin{itemize}
    \item 
$\fdec{\epsilon_{\X}}{\MP(\X)}{\X}$ defines a  deterministic simulation $\fdec{\epsilon_e}{e}{\MP(e)}$ for any $e\colon\X$;
\item
if $\fdec{\tuple{\pi,\ntsigma}}{\MP(d)}{e}$, then $\fdec{\tuple{\pi^\cok,\ntsigma^\cok}}{\MP(d)}{\MP(e)}$;
\item the comonoidal transformation defines a deterministic simulation $\MP(e\otimes d)\lto \MP(e)\otimes \MP(d)$. 
\end{itemize}
Hence we obtain the following theorem. 

\begin{theorem}\label{thm:mpisacomononad}
 $\MP$ defines a comonoidal comonad on the category of empirical models.
\end{theorem}

\subsection{General simulations}\label{ssec:generalsimulations}

The main import of Theorem~\ref{thm:mpisacomononad} is that it allows us to extend our notion of a simulation in a standard manner. Since $\MP$ is comonoidal, its co-Kleisli category 
$\cat{Emp}_{\MP}$ inherits monoidal structure from \cat{Emp}. Since non-contextual empirical models are closed under the product $\otimes$, it is immediate that the following definition results in a category.

\begin{definition}
 Given empirical models $e$ and $d$,
 a \emph{simulation} of $e$ by $d$ is a map $d\otimes c\lto e$ in $\cat{Emp}_{\MP}$,
\ie a map ${\MP(d\otimes c)\lto e}$ in $\cat{Emp}$, for some non-contextual model $c$. We denote the existence of a simulation of $e$ by $d$ as $d\simulates e$, read ``d simulates e''.
\end{definition}

Note that for such maps to compose, it is enough that the class of objects $c$ is coming from is closed under $\otimes$. Hence, we could define even more general simulations by \eg allowing it to range over all quantum-realizable empirical models.

\begin{example}\label{ex:barretpironio}
 Simulations in our sense have been studied less formally in the literature. For instance, \cite[Corollary 2]{barrettpironio2005} shows that any two-output bipartite box can be simulated with PR boxes. In our setting, this means that for any such box $e$ there is an $n$ such that $\PR^{\otimes n}\simulates e$. 
 
 In Theorem 2 they also prove a negative result: there is a quantum realizable 5-partite empirical model $e$ that cannot be simulated by any amount of PR boxes. Strictly speaking, their argument establishes that there is no simulation $\tuple{\pi,\ntsigma}\colon \MP(\PR^{\otimes n}\otimes c)\lto e$ where $c$ is non-contextual and $\tuple{\pi,\ntsigma}$ is of the form $\bigotimes_{i=1}^5 \tuple{\pi_i,\ntsigma_i}$ at the level of measurement scenarios. Restricting to such simulations is operationally motivated for questions of non-locality, since it corresponds to simulations where the PR boxes are distributed among the parties and any party can only access their own PR box (and shared classical randomness).
 
 However, this result readily implies the nonexistence of \emph{any} simulation $PR^{\otimes n}\otimes c\lto e$. This is because PR boxes are bipartite: if, say Alice and Bob both wanted to use the same half of the same PR box, they'd have to make sure that they use the same measurement setting $x_i$. This can of course be coordinated by the classical shared randomness, but this implies that in effect the PR box in question reduces to a submodel where the allowed measurements are $\enset{x_i,y_0,y_1}$ (or a convex combination of such), and these are always non-contextual 
 and thus can be incorporated into the shared classical randomness. Hence $\PR^{\otimes n}\simulates e$ would imply a simulation of the kind that they rule out.
\end{example}

We now establish the equivalence between the categorical view of simulations and the algebraic view based on free operations of a resource theory.

\begin{theorem}\label{thm:simsasterms}
Let $e : \X$ and $d : \Y$ be empirical models.
Then $d\simulates e$ if and only if
there is a typed term $v:\Y \vdash t: \X$ such that $t[d/v] \simeq e$.
\end{theorem}
\begin{proof}
    Suppose that $d\simulates e$.
    Then there is a deterministic simulation
    \[\fdec{\tuple{\pi,\ntsigma}}{\MP(d\otimes c)}{e}\Mdot\]
    But a deterministic simulation amounts to a combination of a coarse-graining and a measurement translation, so that $e=(\pi^*\MP(d\otimes c)/\ntsigma)$. Thus it suffices to show that there is a term $t'$ such that $t'[d/v]\simeq \MP(d\otimes c)$. Since $c$ is non-contextual it can be represented by a closed term by Proposition~\ref{prop:novarterm}. Moreover, it is clear that all of the measurement protocols can be built using conditional measurements repeatedly, establishing the implication from left to right.
    
    
    For the other direction, one possibility would be to prove this by induction on the rules of Table~\ref{table:ops}.
    However, we will be able, more directly, to use the normal form of Proposition~\ref{prop:normalform}.
    Suppose that $v:\Y \vdash t: \X$ and $t[d/v] = e$. We can rewrite $t$ to $t_0$ as in Proposition~\ref{prop:normalform}
    using the equational theory. By Proposition~\ref{prop:soundness}, this rewriting is sound, and therefore ${t_0[d/v] \simeq e}$. The term  $t_0$ is of the form
    \[t_0 = t^1 +_{\lambda_1} (t^2 +_{\lambda_2} + (\cdots +_{\lambda_{n-1}} t^n)\Mcomma\]
    where the terms $t^i$ do not have the probabilistic mixing or choice operations.
    Note that the variable $d$ can only occur in one of these terms,
    as it can occur at most once in any term. Without loss of generality, say that it appears in $t^1$ -- if it appears in another we can use the `commutativity' equation \eqref{eq:convexcomm} of convex combinations to shift it to the first position. Since the remainder of the term has no variables, by Proposition~\ref{prop:novarterm} it represents a non-contextual model. 
    We thus have that $t_0[d/v] \simeq t^1[d/v] +_\lambda c$
    where $c$ is non-contextual.
    We have that $t^1[d/v] \simulates t_0[d/v] \simeq e$ by simulating classical noise through a measurement protocol using an auxilliary non-contextual model.
    
    It is now enough to show that $d \simulates t^1[d/v]$.
    Now, we know that $t^1$ is of the form $t_1$ from Proposition~\ref{prop:normalform}, \ie it is 
    $(f^*t_2)/h$ for some $f$ and $h$ and a term $t_2$ without mixing, translation of measurements, coarse-graining of outcomes, or choice operations. We can read directly a simulation $\fdec{\tuple{f,h}}{t_2[d/v]}{t_1[d/v]}$, so we again reduce the problem to proving that $d \simulates t^2[d/v]$.
    
    Now, recall from Proposition~\ref{prop:normalform} that $t_2$
    is of the form $t_3[x_1?y_1]\cdots[x_n?y_n]$ and $t_3$ is a product of base cases. Since $t_3$ has (at most) one variable $v$, we
    know, again using Proposition~\ref{prop:novarterm},  that $t_3[d/v] \simeq d \otimes c$ with $c$ a non-contextual model. 
    Hence, our goal is to prove that 
    $d \simulates (d \otimes c)[x_1?y_1]\cdots[x_n?y_n]$.
    We can prove the required result by building a deterministic simulation
    \[\MP(d \otimes c) \lto (d \otimes c)[x_1?y_1]\cdots[x_n?y_n]\]
    where the conditional measurements are simulated by protocols (of length at most $n$).
\end{proof}

\begin{theorem}\label{thm:ncfismonotone}
   $d\simulates e$ implies $\NCF(d)\leq \NCF(e)$.
\end{theorem}
\begin{proof}
     This is not hard to prove directly, but we reduce it to earlier results. By Theorem~\ref{thm:simsasterms}, $d\simulates e$ implies that $e\simeq t[d/v]$,
     whence Proposition~\ref{prop:ncfandoperations} implies that $\NCF(d)\leq\NCF(t[d/v])$ by an easy induction on the structure of terms.
\end{proof}

\subsection{No-cloning theorem}\label{ssec:nocloning}

\begin{theorem}[No-cloning]
$e\simulates e\otimes e$ with a non-contextual $c$ if and only if $e$ is non-contextual.
\end{theorem}

\begin{proof}
	  If $e$ is non-contextual, we use it as a free resource and consider \eg the simulation $\fdec{\epsilon_{e\otimes e}}{\MP(e\otimes e)}{e\otimes e}$, proving the implication from right to left.

	 For the other direction, assume that $e\simulates e\otimes e$. We will prove that $\NCF(e)=1$ by first showing that $\NCF(e)>0$ and then that $0<\NCF(e)<1$ is impossible.

	 To see that $\NCF(e)>0$, note that $e\simulates e\otimes e$ implies $e\simulates e^{\otimes n}$ for any $n$. For each $n$, we therefore have a deterministic simulation $T(e\otimes c_n)\lto e^{\otimes n}$ for some non-contextual model $c_n$. Since $c_n$ is non-contextual, we may assume that the simplicial complex it is defined over is a simplex (\ie all faces are measurable), and hence (since having more information than is needed can't hurt) we may assume $c_n$ is defined over a singleton. Thus we can assume that each measurement protocol in the image of the simplicial map starts by first measuring $c_n$ and then proceeds solely in $e$. Then, $c_n$ in effect randomizes which deterministic map $\MP(\X)\lto\X^{\otimes n}$ to use. Consider the underlying simplicial map $\fdec{\pi}{X^{*n}}{\MP(X)}$. It is determined by its components ${\pi_i\colon X\longhookrightarrow \X^{*n} \lto \MP(X)}$. Since there are only finitely many of these, for large enough $n$, we can force at least $k \defeq \left\vert\MP(X)\right\vert$ of these components to agree for any given outcome of $c_n$ -- of course, which components agree might depend on the outcome of $c_n$. For such an $n$, choose an outcome $\lambda\in\supp (c_n)$, and let $i_1,\dots, i_k$ be the indices for which $\pi_i$ agrees given $\lambda$. Then we have a simplicial map $X_{i_1}*\dots*X_{i_k}\lto \MP(X)$ where each component agrees, so by choice of $k$ its image has to lie in a compatible subset  ${U\subseteq\MP(X)}$. This implies that $\MP(e)\hookrightarrow U$ is non-contextual. But then the composed simulation ${\MP(e\otimes c_n)\lto e^{\otimes n}\lto e}$ shows that $e$ has a non-contextual component of size $c_n(\lambda)>0$, namely the pushforward of $\MP(e)\hookrightarrow U$ along the composed simulation. Therefore, $\NCF(e)>0$, as desired. 


	 It remains to show that $0<\NCF(e)<1<$ is impossible, which we do by producing a contradiction. By Proposition~\ref{prop:ncfandoperations} and  $0<\NCF(e)<1$, we have $\NCF(\MP(e\otimes x))=\NCF(e)>\NCF(e)^2=\NCF(e\otimes e)$, while Theorem~\ref{thm:ncfismonotone} implies that $\NCF(e)\leq \NCF(e)^2$.
\end{proof}

Note that the usual no-cloning theorem in quantum mechanics \cite{wootters1982single} states the impossibility of a single quantum procedure cloning arbitrary quantum states, whereas the theorem above states the impossibility of a classical procedure cloning the (outcome statistics) of a known non-contextual system.
One might therefore want to use a different name to avoid confusion.
However, this property is also called no-cloning in the context of a general framework for resource theories developed in~\cite{coecke2016mathematical}.
A similar theorem is proved in the bipartite setting in~\cite{joshi2011nobroadcasting}. Besides restricting to the bipartite setting, a key difference is that their notion of a transformation between empirical models is a priori different: our transformations are operationally motivated and defined, whereas they allow for any transformations satisfying a few reasonable properties such as convexity.

\section{Outlook}\label{sec:outlook}

We have presented a comonadic formulation of simulation of empirical models, and shown that it coincides with a resource theory approach based on free operations. This provides a general and  mathematically well-structured approach to simulation and convertibility, which subsumes the concrete examples studied in the literature, e.g.~\cite{barrettpironio2005,barrettetal2005,jonesmasanes2005interconversions,dupuis2007nouniversalbox,allcock2009closedsets,forsterwolf2011bipartite}. We believe that a robust, general framework of this kind is particularly important for proving \emph{non}-simulability 
results, some concrete examples of which include  \cite{barrettpironio2005,dupuis2007nouniversalbox,forsterwolf2011bipartite}, since it allows for new tools to be used in proving general results. A case in point is the no-cloning theorem above, the proof of which uses basic facts about both simplicial maps and the non-contextual fraction.

An important feature of our approach is that it allows for adaptive protocols. Adaptivity is known to increase expressive power~\cite{forsterwolf2011bipartite}, and it is also important in relation to the measurement-based paradigm for quantum computing (MBQC).

We see this work as providing some important tools for gaining a  deeper understanding of convertibility between empirical models.
There are many further avenues of research to pursue. On the side of applications, one promising direction is to study computation in the MBQC paradigm.
As another application, let us mention the fact that a \emph{possibilistic} empirical model 
corresponds to a constraint-satisfaction problem (CSP) \cite{ab}. Understanding how simulations in our sense relate to notions of reduction between CSPs  might result in novel concepts for CSPs or alternatively, allowing the use of CSP techniques such as pp-definability
to prove negative results about simulations between empirical models.


Another promising direction is to note that the $\MP$ comonad is in fact naturally \emph{graded}, where the grading is by the auxiliary resource used in the simulation. This provides a natural setting for ``relative simulatability'', in which we can ask which (possibly super-quantum) resources can  be converted to which other ones. Ideally, one would show that the no-cloning result holds for any reasonable class of free empirical models.
There is also a natural grading by the length of the allowed runs in the protocols, which may be useful for fine-grained expressiveness results.

One can also ask whether Theorem~\ref{thm:simsasterms} can be strengthened to a bijection between (suitable) simulations and terms up to equality.

Finally, there are several alternative approaches to contextuality, for example, based on operational equivalence~\cite{spekkens2005contextuality}, graph theory~\cite{csw2014graphtheoretic}, hypergraphs~\cite{acin2015combinatorial} or effect algebras~\cite{sander2015effect}.  Each of these can probably accommodate a notion of deterministic map. If the $\MP$ comonad and $\otimes$ can be made to work in these frameworks as well, one might hope to compare the resulting categories -- if they turned out to be equivalent, one could then use tools specific to any of these frameworks when studying simulability.

\section*{Acknowledgment}

Funding from the following is gratefully acknowledged:
Engineering and Physical Sciences Research Council, EP/N018745/1, `Contextuality as a Resource in Quantum Computation' (SA \& RSB); 
the Osk.~Huttunen Foundation (MK); and
the European Union’s Horizon 2020 Research and Innovation Programme under the Marie Sk\l{}odowska-Curie Grant Agreement No.~750523, `Resource Sensitive Quantum Computing' (SM).

\bibliographystyle{IEEEtran}
\bibliography{lics}
\end{document}